\documentclass[10pt]{article}

\usepackage{graphicx}%
\usepackage{amsmath,amssymb,amsthm,upref,bm}%
\usepackage{labelfig}%
\usepackage{mathrsfs} 

\DeclareMathAlphabet{\varmathbb}{U}{pxsyb}{m}{n}

\newtheorem{proposition}{Proposition}%

\newcommand{\D}{\mathrm{d}\kern0.2pt}%
\newcommand{\ii}{\kern0.05em\mathrm{i}\kern0.05em}% 
\newcommand{\RR}{\mathbb{R}}%

\begin{document}

\baselineskip=4.4mm

\makeatletter

\title{\bf Modified Babenko's equation for periodic gravity waves on water of finite
depth}

\author{Evgueni Dinvay$^1$ and Nikolay Kuznetsov$^2$}

\date{}

\maketitle

\vspace{-10mm}

\begin{center}
$^1$Department of Mathematics, University of Bergen, All\'egaten 41, N-5020 Bergen,
\\ $^2$Laboratory for Mathematical Modelling of Wave Phenomena, \\ Institute for
Problems in Mechanical Engineering, Russian Academy of Sciences, \\ V.O., Bol'shoy
pr. 61, St. Petersburg 199178, Russian Federation \\ E-mail:
Evgueni.Dinvay@uib.no\,; nikolay.g.kuznetsov@gmail.com
\end{center}

\begin{abstract}
\noindent A new operator equation for periodic gravity waves on water of finite
depth is derived and investigated; it is equivalent to Babenko's equation considered
in \cite{KD}. Both operators in the proposed equation are nonlinear and depend on
the parameter equal to the mean depth of water, whereas each solution defines a
parametric representation for a symmetric free surface profile. The latter is a
component of a solution of the two-dimensional, nonlinear problem describing steady
waves propagating in the absence of surface tension. Bifurcation curves (including a
branching one) are obtained numerically for solutions of the new equation; they are
compared with known results.
\end{abstract}

\setcounter{equation}{0}

\section{Introduction}

In this note, we consider the nonlinear problem describing steady, periodic waves on
water of finite depth in the absence of surface tension. In its simplest form, this
problem concerns the two-dimensional, irrotational motion of an inviscid,
incompressible, heavy fluid, say water, bounded above by a free surface and below by
a rigid horizontal bottom. It is convenient to formulate the problem in the
following non-dimensional form:
\begin{eqnarray}
&& \psi_{xx} + \psi_{yy} = 0, \quad (x,y) \in D; \label{lapp} \\ && \psi (x, -h) =
-Q , \quad x \in \RR; \label{bcp} \\ && \psi (x, \eta (x)) = 0, \quad x \in \RR;
\label{kcp} \\ && |\nabla \psi (x, \eta (x))|^2 + 2 \eta (x) = \mu , \quad x \in \RR . 
\label{bep}
\end{eqnarray}
Here $D = \{ x \in \RR, -h < y < \eta (x) \}$ is the longitudinal section of the
water domain, say infinite channel of uniform rectangular cross-section; $Q$ (the
rate of flow per channel's unit span) and $h$ (the mean depth of flow) are given
positive constants, whereas $\mu$, $\eta$ and $\psi$ must be found from relations
(\ref{lapp})--(\ref{bep}) so that the constant $\mu$ is positive and the $2
\pi$-periodic and even function $\eta (x)$ (the free surface profile) is
continuously differentiable and satisfies the condition
\begin{equation}
\int_{-\pi}^\pi \eta (x) \, \D x = 0 . \label{eta}
\end{equation}
Furthermore, $\psi (x, y)$ (the stream function) is a $2 \pi$-periodic and even
function of $x$ belonging to the class $C^1 (\bar{D}) \cap C^2 (D)$.

In \cite{KD}, section 2.2, one finds a procedure describing how the above statement
arises from that formulated in terms of dimensional variables and parameters (see,
for example, \cite{Ben}, pp.~340--341, for the latter statement). According to this
procedure, $\mu = 2 \pi c^2 / (g \ell)$ is the Froude number squared with
$g$\,---\,the acceleration due to gravity, $c$\,---\,the mean velocity of flow and
$\ell$\,---\,the wavelength (the smallest period of $\eta$). Moreover, $\mu / 2$ is
the upper bound for $\eta$; it is independent of $h$ and the equality is achieved
only for the wave with the Lipschitz crest; see \cite{CN}.

To explain the aim of this paper, let us outline some previous results on the
problem under consideration. In the first of two brief notes \cite{B,BB} published
shortly before his death, Babenko derived the pseudo-differential operator equation
now named after him. This equation is equivalent to the version of problem
\eqref{lapp}--\eqref{eta} describing waves on infinitely deep water in the following
sense. Every problem's solution defines a solution of the operator equation and vise
versa, to each even equation's solution there corresponds a symmetric solution of
the free-boundary problem. The second note \cite{BB} deals with the existence of
small solutions to the operator equation bifurcating from the trivial one. These
results are also outlined in the book \cite{OS}, section 3.7. Babenko's equation can
be written in the following equivalent form
\begin{equation}
\mu \, \mathcal{C} (v') = v + v \, \mathcal{C} (v') + \mathcal{C} (v' v) , \quad t
\in (- \pi, \pi) \, , \label{bid}
\end{equation}
which was studied in detail by Buffoni, Dancer and Toland in their comprehensive
articles \cite{BDT1,BDT2}. Here, the parameter $\mu$ (the same as in the Bernoulli
equation) must be found along with an even, $2 \pi$-periodic function $v (t)$ from
the Sobolev space $W^{1,2} (-\pi, \pi)$; the prime $'$ denotes differentiation with
respect to $t$ and ${\cal C}$ is the $2 \pi$-periodic Hilbert transform; see, for
example, \cite{Z}.

In our previous article \cite{KD}, we extended the approach developed in \cite{B} to
the case when the water has finite depth by deducing a pseudo-differential equation
analogous to \eqref{bid} and equivalent to problem \eqref{lapp}--\eqref{eta} in the
same sense as described above. The equation obtained in \cite{KD} is as follows:
\begin{equation}
\mu \, \mathcal{B}_r (v') = v + v \, \mathcal{B}_r (v') + \mathcal{B}_r (v' v) , \quad t
\in (- \pi, \pi) \, . \label{37}
\end{equation}
Here the operator $\mathcal{B}_r$, $r \in (0, 1)$, is defined on the space
$L^2_{per} (-\pi, \pi)$ of periodic functions by linearity from the relations
\begin{equation}
\mathcal{B}_r (\cos n t) = \frac{1 + r^{2 n}}{1 - r^{2 n}} \sin n t \ \ \mbox{for} \
n \geq 0 , \quad \mathcal{B}_r (\sin n t) = - \frac{1 + r^{2 n}}{1 - r^{2 n}} \cos n
t \ \ \mbox{for} \ n \geq 1 , \label{HTB_r}
\end{equation}
and so $\mathcal{B}_r$ is similar to the $2 \pi$-periodic Hilbert transform ${\cal
C}$ defined by these formulae with $r=0$; see \cite{Z}. Hence \eqref{37} is
analogous to \eqref{bid} and it is natural to refer to \eqref{37} as Babenko's
equation for water of finite depth. In this equation, the parameter $r$ arises from
some auxiliary conformal mapping (see \cite{KD}, section 3.1) and can be referred to
as the conformal radius of $D$.

Let us specify the mentioned equivalence of problem (\ref{lapp})--(\ref{eta}) and
equation \eqref{37}. If $(\mu, \eta, \psi)$ is a solution of the problem, then it
defines some $r$ and $v$ so that $(\mu, v)$ satisfies \eqref{37}. On the contrary,
if $(\mu, v)$ is a solution of \eqref{37} for some $r$ with $2 \pi$-periodic and
even $v$, then there exist $h$, $Q$ and $(\eta, \psi)$ such that $(\mu, \eta, \psi)$
is a solution of (\ref{lapp})--(\ref{eta}) with $h$ as the mean depth of the water
domain $D$ and $-Q$ standing on the right-hand side of \eqref{bcp}. However, it is
not clear how $h$ corresponds to $r$ and this is a drawback of equation \eqref{37}.

Another version of Babenko's equation for water of finite depth was obtained by
Constan\-tin, Strauss and V\u{a}rv\u{a}ruc\u{a} \cite{CSV} (see Remark~4 in their
article). It involves a pseudo-differential operator whose definition is similar to
\eqref{HTB_r} with another multiplier replacing $(1 + r^{2 n}) / (1 - r^{2 n})$.
Instead of $r$ on which $\mathcal{B}_r$ depends, the alternative operator depends on
the so-called conformal depth. Like $r$, this parameter has no direct physical
interpretation.

Therefore, the aim of this note is to propose and analyse an operator equation
equivalent to \eqref{37}, but having the advantage that the operators involved in it
depend on the mean depth of water $h$. Of course, the cost of this advantage is that
the operators in the new equation are nonlinear, whereas $\mathcal{B}_r \, \D / \D
t$ is linear as well as the operator used in the equation obtained in \cite{CSV}.
Our new equation has the form similar to that of \eqref{37} and is derived with the
help of a transformation based on the spectral decomposition of the self-adjoint
operator $\mathcal{B}_r \, \D / \D t$; see (4.3) in \cite{KD} and
\eqref{inverse_spectralBabenko} below.

\vspace{-2mm}

\section{Modified Babenko's equation}

We begin with reminding an assertion obtained in \cite{KD}, section~4.1; $v \in
W^{1,2} (-\pi, \pi)$ is an even, $2 \pi$-periodic solution of \eqref{37} if and only
if $w \in W^{1,2} (0, \pi)$\,---\,the restriction of $v$ to $(0,
\pi)$\,---\,satisfies the equation\\[-3mm]
\begin{equation}
\mu \mathcal J_r w = w + w \mathcal J_r w + \frac 12 \mathcal J_r (w^2) \, , \ \
\mbox{where} \ \ \mathcal J_r = \sum_{n = 1}^{\infty} \lambda_n P_n \, .
\label{spectralBabenko}
\end{equation}
Here $P_n$ is the projector onto the subspace of $L^2 (0, \pi)$ spanned by $\cos
nt$, $n = 0,1,\dots$, and
\begin{equation}
\lambda_n = n \frac{ 1 + r^{2n} }{ 1 - r^{2n} } \, , \quad n = 1,2,\dots \, ,
\label{lambda_n}
\end{equation}
is the corresponding eigenvalue of $B_r \D / \D t$.

The equivalence of \eqref{spectralBabenko} and \eqref{37} follows from the fact that
$\mathcal J _r w = \mathcal B_r (v')$ almost everywhere on $(0, \pi)$, which is a
consequence of the spectral decomposition of $\mathcal B_r \D / \D t$. Indeed, the
latter operator (it is present in all terms of \eqref{37} except for the first one
on the right-hand side) is self-adjoint on $L^2_{per} (-\pi, \pi)$ and its sequence
of eigenvalues is \eqref{lambda_n}.

Another equivalent form of equation \eqref{37} was proposed in \cite{KD} (see
section~4.1 of this paper), for which purpose the bounded operator\\[-3mm]
\[ \mathcal L_r = \sum_{n = 0}^{\infty} \mu_n P_n
\]
was introduced. Here $\mu_0 = 1$ and\\[-3mm]
\begin{equation}
\mu_n = \lambda_n^{-1} = \frac{1 - r^{2n}}{n (1 + r^{2n})} \quad \mbox{for} \ 
n=1,2,\dots \, .
\label{mu_n}
\end{equation}
Hence $\mathcal L_r$ is invertible and ${\mathcal L_r}^{-1} = P_0 + \mathcal J_r$;
that is, $\mathcal L_r \mathcal J_r = I - P_0$, where $I$ is the identity operator.
Applying $\mathcal L_r$ to both sides of \eqref{spectralBabenko}, one obtains
\begin{equation}
\mu (I - P_0) w = \mathcal L_r w + \mathcal L_r ( w \mathcal J_r w ) + \frac 12 (I -
P_0) w^2 \, , 
\label{inverse_spectralBabenko}
\end{equation}
which is equivalent to \eqref{spectralBabenko}, and so to \eqref{37}. The advantage
of \eqref{inverse_spectralBabenko} is that the unbounded operator $\mathcal J_r$
appears only in the middle term on the right-hand side. This version of Babenko's
equation is essential for our further considerations.

\subsection{Derivation of the modified Babenko's equation}

First, we show that $h$ can be used as a parameter replacing $r$ in the operators
analogous to $\mathcal L_r$ and $\mathcal J_r$; they are involved in the modified
version of equation \eqref{inverse_spectralBabenko} whose solution we also denote
by $w$ for the reason that will become clear below.

Let $v$ be the even extension from $(0, \pi)$ to $(-\pi, \pi)$ of a solution to
\eqref{inverse_spectralBabenko}, and so $v$ solves \eqref{37}. Let $\{ b_0,\ b_1,\
\dots \} \subset \RR$ denote the sequence of its modified Fourier coefficients
uniquely defined by the expansion\\[-3mm]
\begin{equation}
v(t) = b_0 + \sum_{k=1}^\infty b_k \left( 1 - r^{2 k} \right) \cos kt \, , \quad t
\in (-\pi, \pi) \, .
\label{first}
\end{equation}
It should be noted that $b_0 = P_0 w$, whereas there is no such a simple relation
between $b_n$ and $P_n w$ for $n \geq 1$. Hence, it is convenient to use the same
notation for the mean value of a constant function and its projection.

On an auxiliary $u$-plane, we consider the following annular domain with a cut:
\[ A_r = \{ r < |u| < 1 ; \ \Re\,u \notin (-1 , -r) \ \mbox{when} \ \Im\,u =0 \} \]
(see \cite{KD}, Fig.~2). Using the coefficients in \eqref{first} we
define\\[-3mm]
\begin{equation} 
z (u) = \ii \Big[ \log u + b_0 + \sum_{k=1}^\infty b_k \left( u^k - r^{2 k} u^{-k}
\right) \Big] \, , \label{z_u'}
\end{equation}
which is a holomorphic function on $A_r$. It maps $A_r$ conformally onto a wave-like
domain $D_{2 \pi}$ on the $z$-plane. (This is a consequence of the boundary
correspondence principle; see, for example, \cite{E}, ch.~5, Theorem~1.3, whereas
figures~2--6 in \cite{KD}, Section~3.3, illustrate how to check this principle
numerically.) The horizontal extent of this domain is $2 \pi$ and\\[-6mm]
\begin{eqnarray}
&& \ \ \ \ \ \ \ \ \ z (e^{it}) = x(t) + \ii v(t) , \quad t \in (-\pi, \pi) ,
\label{surface_parametrization} \\ &&
x (t) = - t - \sum_{k=1}^\infty b_k \left( 1 + r^{2 k} \right) \sin kt = -t -
(\mathcal B_r v) (t) 
\label{x_t}
\end{eqnarray}
is the parametric representation of its upper profile. Putting $u=r$ in
\eqref{z_u'}, we see that the result has a constant imaginary part, say $-\hslash$,
which gives the level of the horizontal bottom of $D_{2 \pi}$:
\begin{equation}
\hslash = - b_0 - \log r . \label{hslash}
\end{equation}

Let us show that $\hslash$ is equal to some $h > 0$ used in the formulation of
problem (\ref{lapp})--(\ref{eta}), which means that $D_{2 \pi}$ has the same depth
as $D$. Indeed, this is the case when the free surface profile given by
\eqref{surface_parametrization} and \eqref{x_t} has the zero mean value and $\hslash
> 0$. Our proof of these facts we begin by noticing that $P_0 \mathcal J_r = 0$,
which follows from the definition of $\mathcal J_r$. Moreover, for solutions of
\eqref{inverse_spectralBabenko} we have the following equality:
\begin{equation}
P_0 (w + w \mathcal J_r w) = 0 . \label{P_0}
\end{equation}
Since $w$ is also a solution of \eqref{spectralBabenko}, we see that \eqref{P_0}
follows by applying $P_0$ to both sides of \eqref{spectralBabenko}. Furthermore, $v$
is the even extension of $w$ and $\mathcal J _r w = \mathcal B_r v'$ almost
everywhere on $(0, \pi)$. Therefore, \eqref{P_0} can be written in the form
\begin{equation}
\int_0^{\pi} [ v + v \mathcal B_r (v') ] \, \D t = 0 \quad  \Longleftrightarrow
\quad \int_{-\pi}^{\pi} [ v + v \mathcal B_r (v') ] \, \D t = 0 \, .
\label{mean_value}
\end{equation}
Here, the last equality is a consequence of the fact that $\mathcal B_r (v')$ is
even when $v$ is even. Now, formula \eqref{x_t} implies that the second integral in
\eqref{mean_value} is equal to
\[ -\int_{-\pi}^{\pi} v (t) x' (t) \D t \, ,
\]
and so the free surface profile has the zero mean value.

It remains to prove that $P_0 w \leq 0$ for solutions of
\eqref{inverse_spectralBabenko}; indeed, in view of \eqref{first} and \eqref{hslash}
this inequality implies that $\hslash = - [ \log r + P_0 w] > 0$ provided $0 < r <
1$. According to \eqref{P_0}, the inequality $P_0 w \leq 0$ is equivalent to the
following one:
\[ P_0 (w \mathcal J_r w) = \frac{1}{\pi} \int_0^\pi w (t) (\mathcal J_r w) (t) \, 
\D t \geq 0 \, .
\]
Since $\mathcal J_r$ is a positive definite operator (this follows from the fact
that all its eigenvalues are positive), the last integral is positive which
completes the proof.

From now on, we write $h$ instead of $\hslash$ and consider it as the parameter on
which will depend the operators we are going to defined. First, we consider the
nonlinear functional
\begin{equation}
r_h (w) = \exp \{- h - P_0 w \}  \label{r_h}
\end{equation}
(cf. \eqref{hslash} above). Changing $r$ to this functional in formula
\eqref{lambda_n} (it gives the sequence of eigenvalues of $\mathcal J_r$), we obtain
the following functionals
\begin{equation}
\lambda_n^{(h)} (w) = n \frac{ 1 + [r_h (w)]^{2n} }{ 1 - [r_h (w)]^{2n} } , \quad n
= 1,2,\dots \, , \label{lambda_n_w}
\end{equation}
all of which are well defined provided $P_0 w \neq - h$. Changing $\{ \lambda_n
\}_{n=1}^\infty$ to these functionals in the definition of $\mathcal J_r$, we
introduce the following nonlinear operator:
\[ \mathcal J_h w = \sum_{n = 1}^{\infty} \Big[ \lambda_n^{(h)} (w) \Big] P_n w , \quad
w \in W^{1,2} (0, \pi) , \ \ P_0 w > - h .
\]
In the same way, we define on $L^2(0, \pi)$ the nonlinear operator:
\[ \mathcal L_h w = P_0 w + \sum_{n = 1}^{\infty} \Big[ \mu_n^{(h)} (w) \Big] P_n w 
\, , \ \ \mbox{where} \ \ \mu_n^{(h)} (w) = \frac{1 - [r_h (w)]^{2n}}{n \{ 1 + [r_h
(w)]^{2n} \}} \, .
\]
It should be noted that
\begin{equation}
\lambda_n^{(h)} (w) = \lambda_n^{(h)} (P_0 w) \quad \mbox{and} \quad \mu_n^{(h)} (w)
= \mu_n^{(h)} (P_0 w) \quad \mbox{for every} \ n = 1,2,\dots \, ,
\label{mu_of_mean_value}    
\end{equation}
because in view of \eqref{r_h} we have that $r_h (w) = r_h (P_0 w)$. Of course, the
first of relations \eqref{mu_of_mean_value} is true only when $P_0 w \neq - h$,
whereas the last one holds for all $w \in L^2 (0, \pi)$.

In terms of operators $\mathcal J_h$ and $\mathcal L_h$ defined for every $h > 0$,
we write down the equation:
\begin{equation}
\mu (1 - P_0) w = \mathcal L_h w - \mathcal L_h ( - w \mathcal J_h w ) + \frac 12 (1
- P_0) (w^2) \, . \label{mod_Babenko}
\end{equation}
Since it is similar to \eqref{inverse_spectralBabenko}, we will refer to
\eqref{mod_Babenko} as the {\it modified Babenko's equation}. The reason for this
becomes clear from the following assertion.

\begin{proposition}
Let $(\mu, w)$ with $\mu > 0$ and $w \in W^{1,2} (0, \pi)$ be a solution of
\eqref{inverse_spectralBabenko} for some fixed value of $r \in (0, 1)$. Then $w$
belongs to the domain of $\mathcal J_h$ with $h = - \log r - P_0 w > 0$ and the pair
$(\mu, w)$ satisfies \eqref{mod_Babenko}.

On the contrary, let $\mu > 0$ and $w \in W^{1,2} (0, \pi)$ solve
\eqref{mod_Babenko} with $h > 0$, and let $P_0 w > - h$. Then $(\mu, w)$ is a
solution of \eqref{inverse_spectralBabenko} with $r = \exp \{- h - P_0 w \} \in (0,
1)$.
\end{proposition}

\begin{proof}
The fact that $h = - \log r - P_0 w > 0$ is already established under the assumption
that $(\mu, w)$ is a solution of equation \eqref{inverse_spectralBabenko}.
Substituting this $h$ into \eqref{r_h}, we see that $r_h (w) = r$, and so for all
$n=1,2,\dots$
\begin{equation}
\lambda_n^{(h)} (w) = \lambda_n \quad \mbox{and} \quad \mu_n^{(h)}(w) = \mu_n 
\label{mu_n*}
\end{equation}
are eigenvalues of the operators $\mathcal J_r$ and $\mathcal L_r$ respectively.
Furthermore, it is clear that $w$ belongs to the domain of $\mathcal J_h$, and so we
have
\[ \mathcal J_h w = \mathcal J_r w \quad \mbox{and} \quad \mathcal L_h w =
\mathcal L_r w
\]
for solutions of \eqref{inverse_spectralBabenko}.

Let us consider the following expression
\begin{equation}
\mu (1 - P_0) w - \mathcal L_h w + \mathcal L_h ( - w \mathcal J_h w ) - \frac 12
(1 - P_0) (w^2) \, ,
\label{expr}
\end{equation}
where $(\mu, w)$ is a solution of equation \eqref{inverse_spectralBabenko}. To prove
the first assertion of Proposition~1, we have to show that this expression vanishes
identically. Taking into account the obtained relations, \eqref{expr} reduces to:
\[ \mu (1 - P_0) w - \mathcal L_r w + \mathcal L_h ( - w \mathcal J_r w ) - \frac 12
(1 - P_0) (w^2) \, .
\]
Thus, to complete the proof we have to show that $\mathcal L_h ( - w \mathcal J_r w)
= - \mathcal L_r ( w \mathcal J_r w )$. Indeed, this is true because
\begin{eqnarray}
\mathcal L_h ( - w \mathcal J_r w ) = P_0 ( - w \mathcal J_r w ) + \sum_{n =
1}^{\infty} \Big[ \mu_n^{(h)} ( - w \mathcal J_r w ) \Big] P_n ( - w \mathcal J_r w)
\nonumber \\ = - P_0 ( w \mathcal J_r w ) + \sum_{n = 1}^{\infty} \mu_n P_n ( - w
\mathcal J_r w ) = - \mathcal L_r ( w \mathcal J_r w ) \, . \label{ch}
\end{eqnarray}
Here the second formula \eqref{mu_n*} and the definition of $\mathcal L_r$ are taken
into account. Substituting this into \eqref{expr}, we see that this expression
vanishes because $(\mu, w)$ is a solution of \eqref{inverse_spectralBabenko}, and so
the proof of the first assertion is complete.

To prove the second assertion we notice that the assumption $P_0 w > - h$ implies
that $r = \exp \{- h - P_0 w \} \in (0, 1)$ and $r_h (w) = r$ in view of
\eqref{r_h}. Therefore, we have
\begin{equation}
\lambda_n^{(h)} (w) = \lambda_n \quad \mbox{and} \quad \mu_n^{(h)}(w) = \mu_n \quad
\mbox{for all} \ n=1,2,\dots ,
\label{mu_n'}
\end{equation}
provided $w$ is a solution of \eqref{mod_Babenko}, and so $\mathcal J_h w = \mathcal
J_r w$ and $\mathcal L_h w = \mathcal L_r w$. (It should be emphasised that formulae
\eqref{mu_n'} and \eqref{mu_n*} only look exactly the same, but $w$ denotes a
solution of \eqref{inverse_spectralBabenko} in \eqref{mu_n*}, whereas $w$ is a
solution of \eqref{mod_Babenko} in \eqref{mu_n'}.) Hence equation
\eqref{mod_Babenko} reduces to
\begin{equation}
\mu (1 - P_0) w = \mathcal L_r w - \mathcal L_h ( - w \mathcal J_r w ) + \frac 12
(1 - P_0) (w^2) \, , \quad t \in (0, \pi) \, . \label{1_May}
\end{equation}
Applying $P_0$ to both sides of \eqref{1_May}, we obtain $P_0 w = P_0 ( - w \mathcal
J_r w )$. Using this and taking into account the chain of equalities \eqref{ch}
with $w$ denoting a solution of \eqref{mod_Babenko} (here the second formula
\eqref{mu_n'} is used along with the definition of $\mathcal L_r$), relation
\eqref{1_May} turns into equation \eqref{inverse_spectralBabenko} which completes
the proof of the second assertion.
\end{proof}

Small solutions of \eqref{mod_Babenko} exist because this equation is equivalent to
\eqref{inverse_spectralBabenko} which, in its turn, is equivalent to \eqref{37}. For
the latter equation, the existence of small solutions was established directly with
the help of the Crandall--Rabinowitz theorem (see \cite{CR}, Theorem~2). Therefore,
we restrict ourselves just to formulating the following assertion; see further
details in \cite{KD}, Section~3.2.

\begin{proposition}
For every $n=1,2,\dots$ there exists $\varepsilon_n > 0$ such that for $0 < |s| <
\varepsilon_n$ there is a family $\big( \mu_n^{(s)}, \, w_n^{(s)} \big)$ of
solutions to equation \eqref{mod_Babenko}. Together with the bifurcation point
$(\mu_n , 0)$, where $\mu_n$ is given by \eqref{mu_n}, the points of this family
form the continuous curve
\[ C_n = \big\{ \big( \mu_n^{(s)}, \, w_n^{(s)} (t) \big) : |s| < \varepsilon_n
\big\} \subset \RR \times W^{1,2} (0, \pi) , \quad n=1,2,\dots \, .
\]
Moreover, the asymptotic formulae
\begin{equation}
\mu_n^{(s)} = \mu_n + o (s) \, , \quad w_n^{(s)} (t) = s \cos n t + o (s)
\label{s_a}
\end{equation}
hold for these solutions as $|s| \to 0$. Finally, each curve $C_n$ is of class
$C^1$.
\end{proposition}

\subsection{Solutions of equation (23) define periodic waves}

Let us outline how to recover a solution of problem \eqref{lapp}--\eqref{bep} from
the pair $(\mu, w)$ that solves equation \eqref{mod_Babenko} for some $h > 0$.
According to Proposition~1, $(\mu, w)$ is also a solution of \eqref{spectralBabenko}
with $r = \exp \{- h - P_0 w \} \in (0, 1)$, and so $(\mu, v)$ is a solution of
\eqref{37} provided $v$ is the even extension of $w$ from $(0, \pi)$ to $(-\pi,
\pi)$. Then the considerations expounded at the beginning of Section~2.1 are
applicable to $v$, thus giving the parametric curve
\[ \{ x = - t - (\mathcal{B}_r \, v) (t) , \ \ y = v (t) ; \ \ t \in 
[-\pi, \pi] \} \, .
\]
It has the zero mean value, and so serves as the upper part of the boundary for the
wave domain $D_{2 \pi}$ (therefore, a part of free surface profile $\eta$). The
horizontal extent of $D_{2 \pi}$ is a wave period equal to $2 \pi$, whereas its
bottom is
\[ \{ x \in [-\pi, \pi] , \ \ y = -h \} \, .
\]

The next step is to show that $D_{2 \pi}$ serves as a part of $D$ for some periodic
wave; that is, there exists a $2 \pi$-periodic stream function $\psi$, which is
defined on $\overline {D_{2 \pi}}$ and satisfies conditions \eqref{bcp}--\eqref{bep}
with some constant serving as the right-hand side term in \eqref{bcp}, whereas $\mu$
stands in \eqref{bep}. To show the existence of $\psi$ one has to repeat literally
the considerations presented in \cite{KD} at the end of Section~3.3.

\section{Numerical results}

In this section, we describe a method aimed at solving equation \eqref{mod_Babenko}
numerically. Being based on calculation of the solution's Fourier coefficients $b_0,
b_1, \ldots$, it is similar to that applied for the numerical solution of
\eqref{inverse_spectralBabenko}; see \cite{KD}, Section~4. Namely, a modified
version of the software SpecTraVVave is applied (the latter is available freely at
the site cited as \cite{MVK}; its detailed description can be found in \cite{KMV}).

\subsection{Discretisation of equation (23)}

We apply the standard cosine collocation method in the same way as in \cite{KD}:
solutions are sought as linear combinations of $\cos mx$, $m = 0, 1, \dots$, which
form a basis in $L^2(0, \pi)$. For the discretisation we use the subspace $\mathcal
S_N$ spanned by the first $N$ cosines which are defined by their values at the
collocation points $x_n = \pi \frac{2n - 1}{2N}$, $n = 1, \ldots, N$. Thus, $f \in
W^{1,2} (0, \pi)$ is represented by the vector $f^N$ given by its coordinates
\[ f^N_n = \sum_{k=0}^{N-1} (P_k f) (x_n) , \quad n = 1, \dots, N .
\]
The operator $\mathcal L_h$ is discretised as follows:
\[ ( \mathcal L_h^N f^N )_n = \sum_{k=0}^{N-1} (P_k \mathcal L_h f) (x_n) , \quad
n = 1, \dots, N ,
\]
whereas $\mathcal J_h^N$ and $P_0^N$ denote the discretisations of $\mathcal J_h$
and $P_0$ respectively. It should be mentioned that $\mathcal L_h^N$ and $\mathcal
J_h^N$ are nonlinear, and so distinguish essentially from their counterparts in
\cite{KD}. However, this has only an insignificant influence resulting in a slight
increase of the computation time. These definitions are correct because $f^N$
defines the function $f$ with values $f(x_n) = f^N_n$ uniquely up to a projection on
the subspace orthogonal to $\mathcal S_N$.

The discrete analogue of \eqref{mod_Babenko} is as follows:
\begin{equation}
\mathcal L_h^N w^N - \mu \left( 1 - P_0^N \right) w^N - \mathcal L_h^N \left( - w^N
\mathcal J_h^N w^N \right) + \frac 12 \left( 1 - P_0^N \right) \left( w^N \right)^2
= 0 . \label{discrete_Babenko}
\end{equation}
Solutions $(\mu, w^N)$ of this equation constitute curves in the $(\mu, a)$-plane,
where
\[ a = \| w^N \| = \max_n |w^N_n| ,
\]
and we parametrise them in order to make the calculation procedure more effective.
Due to the presence of a new parameter, say $\theta$, we have that $\mu = \mu
(\theta)$ and $a = a (\theta)$. Substituting $\mu (\theta)$ into
\eqref{discrete_Babenko} instead of $\mu$, we complement this algebraic system by
the following equation:
\begin{equation}
\max _{n = 1, \ldots, N} |w^N_n| = a(\theta) . \label{waveheight}
\end{equation}
The resulting system \eqref{discrete_Babenko}--\eqref{waveheight} has $N + 1$
equations with the unknowns $\theta, w^N_1, \ldots, w^N_N$. The standard Newton's
iteration method is applicable for finding branches bifurcating from a trivial
solution, whereas the Crandall--Rabinowitz asymptotic formula \eqref{s_a} allows us
to find an initial approximation. Further details concerning the proposed
parametrisation and the particular realisation of this algorithm can be found in
\cite{KMV}.

\begin{figure}
\vspace{-4mm}
\begin{center}
\SetLabels
 \L (0.5*0.01) $\mu$\\
 \L (0.64*0.42) $C_1$\\
 \L (0.29*0.38) $C_2$\\
 \L (0.3*0.56) \tiny $C_{21}$\\
 \L (-0.05*0.48) $\| w \|_\infty$\\
\endSetLabels 
\leavevmode
\strut\AffixLabels{\includegraphics[width=80mm]{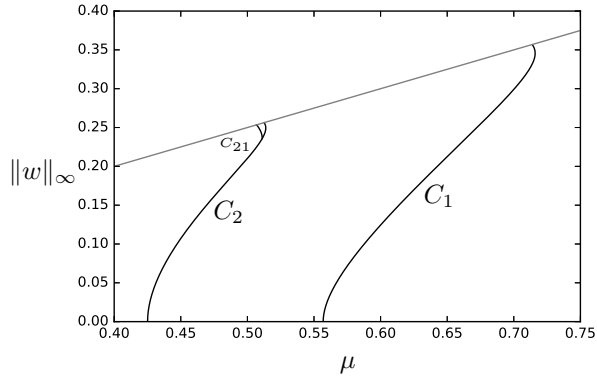}}
\end{center}
\vspace{-4mm}
\caption{The solution branches $C_1$ and $C_2$ for equation \eqref{mod_Babenko} with
$h = \pi /5$, bifurcating from the zero solution at $\mu$ approximately equal to
$0.55689$ and $0.42507$ respectively. The secondary solution branch denoted by
$C_{21}$ bifurcates from $C_2$ at $\mu \approx 0.51113$. The upper bound $\mu / 2$
mentioned prior to equation \eqref{37} is also included.} 
\vspace{-4mm}
\end{figure}

\subsection{Bifurcation curves for equation (23)}

Here we present numerical results obtained for equation \eqref{mod_Babenko} with $h
= \pi / 5$. The solution branches $C_1$ and $C_2$ (the principal ones) are plotted in
Fig.~1 in terms of $\mu$ and the solution norm $\|w\|_\infty$ in the space $L^\infty
(0, \pi)$. Let us describe some characteristics of $C_1$ and $C_2$ and compare them
with properties of solution branches obtained in \cite{CN} and \cite{KD}, where
other equations are used.

Branch $C_1$ bifurcates from the zero solution at $\mu \approx 0.55689$ and
terminates at the solution corresponding to the wave of extreme form for which
$\|w\|_\infty \approx 0.35686$. This branch has no secondary bifurcation points as
the analogous branches for equations \eqref{bid} (see \cite{BDT1,BDT2} for the
rigorous proof and detailed discussion) and \eqref{37} with $r = 4/5$ (see
\cite{KD}, Fig.~1). On the other hand, $C_2$ bifurcating from the zero solution at
$\mu \approx 0.42507$ has one secondary bifurcation point at $\mu \approx 0.51113$;
the same property has $C_3$ for \eqref{37} with $r = 4/5$ (see \cite{KD}, Fig.~9).

\begin{figure}
\begin{center}
\SetLabels
 \L (0.5*0.01) $c$\\
 \L (0.66*0.42) $C_1$\\
 \L (0.25*0.38) $C_2$\\
 \L (0.3*0.59) \tiny $C_{21}$\\
 \L (-0.05*0.48) $\| w \|_\infty$\\
\endSetLabels 
\leavevmode
\strut\AffixLabels{\includegraphics[width=80mm]{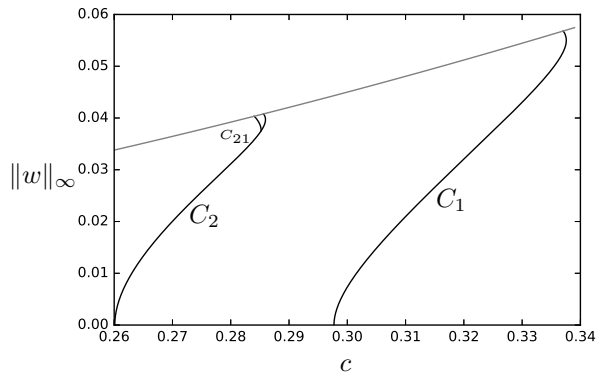}}
\end{center}
\vspace{-4mm} \caption{Plot of the solution branches $C_1$, $C_2$ and $C_{21}$ for
equation \eqref{mod_Babenko} with $h = \pi /5$ in terms of the variables used in
\cite{CN}; namely, the phase velocity $c$ and $\| w \|_\infty$. The upper bound $\mu
/ 2$ mentioned prior to equation \eqref{37} is also included.} \vspace{-4mm}
\end{figure}

\begin{figure}
\vspace{-6mm}
\begin{center}
\SetLabels
 \L (0.48*0.01) $\| w \|_\infty$\\
 \L (0.02*0.5) $r$\\
\endSetLabels 
\leavevmode
\strut\AffixLabels{\includegraphics[width=80mm]{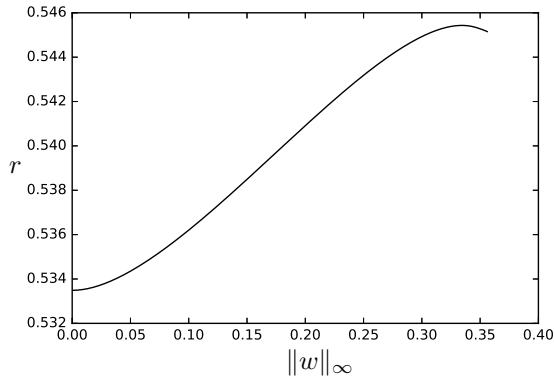}}
\end{center}
\vspace{-4mm} \caption{Plot of values of the parameter $r$ corresponding to
solutions which belong to the branch $C_1$ in Fig.~1 and are calculated for $h =
\pi /5$.} \vspace{-4mm}
\end{figure}

Moreover, both $C_1$ and $C_2$ exhibit the phenomenon of a turning point, occurring
high on each of these branches. The corresponding largest values of $\mu$ are
approximately equal to $0.71604$ for $C_1$ and to $0.51381$ for $C_2$; the
$L^\infty$-norms of the corresponding solutions are approximately equal to $0.34553$
for $C_1$ and to $0.24935$ for $C_2$. The meaning of turning points is that the
fastest traveling waves of given periods correspond to them. This phenomenon is
related to the `Tanaka instability' found numerically by Tanaka \cite{Tan}, and
later investigated analytically in \cite{Saf}. By means of a different method this
property was demonstrated in \cite{CN}, whereas our method shows that it also takes
place for equation \eqref{mod_Babenko} on $C_1$ and $C_2$.

Plot of the bifurcation curves $C_1$, $C_2$ and $C_{21}$ rescaled from our variables
$\mu$ and $\| w \|_\infty$ to those used in \cite{CN} is given in Fig.~2. It
demonstrates a good agreement with the results obtained by virtue of the numerical
technique introduced in \cite{CS} and based on the Taylor expansion of the
Dirichlet--Neumann operator in homogeneous powers of the surface elevation $\eta$.

In Fig.~3, we give a plot of values attained by $r$ when solutions of
\eqref{mod_Babenko} (they are calculated for $h = \pi /5$) run over the branch~$C_1$
shown in Fig.~1. These values obtained with the help of the nonlinear functional
$r_h (w)$ (see formula \eqref{r_h}) have a remarkable feature: the dependence of $r$
on the solution's norm $\| w \|_\infty$ is not monotonic when the latter varies on
the interval $(0, 0.35686)$. The last value is approximately equal to the
$L^\infty$-norm of the solution corresponding to the wave of extreme form. The
single maximum of this curve approximately equal to $0.54543$ is attained at $\| w
\|_\infty \approx 0.33433$. The latter value is slightly less than the norm
corresponding to the turning point solution; namely, $\approx 0.34553$.

\begin{figure}%[b]
%\vspace{-4mm}
\begin{center}
\SetLabels
 \L (0.5*0.01) $\mu$\\
 \L (0.64*0.5) $C_5$\\
 \L (0.29*0.62) $C_{51}$\\
 \L (0.46*0.7) $C_{52}$\\
 \L (0.56*0.72) $C_{53}$\\
 \L (-0.08*0.48) $\| w \|_\infty$\\
\endSetLabels 
\leavevmode
\strut\AffixLabels{\includegraphics[width=80mm]{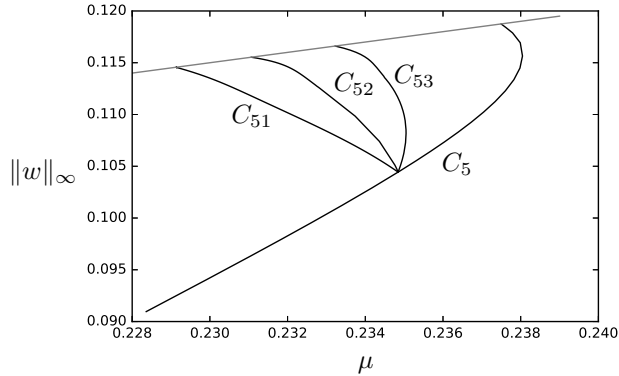}}
\end{center}
\vspace{-5mm} \caption{The upper part of the solution diagram $C_5$ is plotted for
equation \eqref{mod_Babenko} with $h = \pi /5$. Three secondary branches bifurcating
from $C_5$ at $\mu \approx 0.23484$ and $\| w \|_\infty \approx 0.10444$ are denoted
by $C_{51}$, $C_{52}$ and $C_{53}$. The upper bound $\mu / 2$ mentioned prior to
equation \eqref{37} is included.} \vspace{-4mm}
\end{figure}

Now we turn to Fig.~4, where the upper part of the solution diagram $C_5$ is
plotted. This curve bifurcates from the zero solution at $\mu \approx 0.19925$,
whereas there are three secondary branches (denoted by $C_{51}$, $C_{52}$ and
$C_{53}$), which, within the accuracy of our computations, bifurcate from $C_5$ at
the point with $\mu \approx 0.23484$ and $\| w \|_\infty \approx 0.10444$.
Unfortunately, the adopted accuracy is insufficient to check whether there is a
single bifurcation point or there are two or even three of them, but it is extremely
difficult to improve accuracy. Moreover, it should be emphasised that the accurate
method developed by Clamond and Dutykh \cite{CD} for computation of steady waves
does not work for waves having more than one crest per period, and so cannot be
implemented close to bifurcation points. Indeed, there are five crests per period on
every wave profile corresponding to a solution belonging to either of branches
bifurcating from $C_5$; see Figs.~5--7. Moreover, the highly accurate method based
on group theoretic technique, which was applied by Aston \cite{A} for computation of
bifurcation points in the simpler case of the infinitely deep water, yields only two
these points on the fifth branch of solutions. Thus, either the third secondary
branch is absent/lost in \cite{A} or two of these branches bifurcate from a single
point on $C_{5}$ when the water is infinitely deep.

\begin{figure}[b]
\vspace{-4mm}
\begin{center}
\SetLabels
 \L (0.5*0.01) $x$\\
 \L (-0.005*0.48) $\eta$\\
\endSetLabels 
\leavevmode
\strut\AffixLabels{\includegraphics[width=80mm]{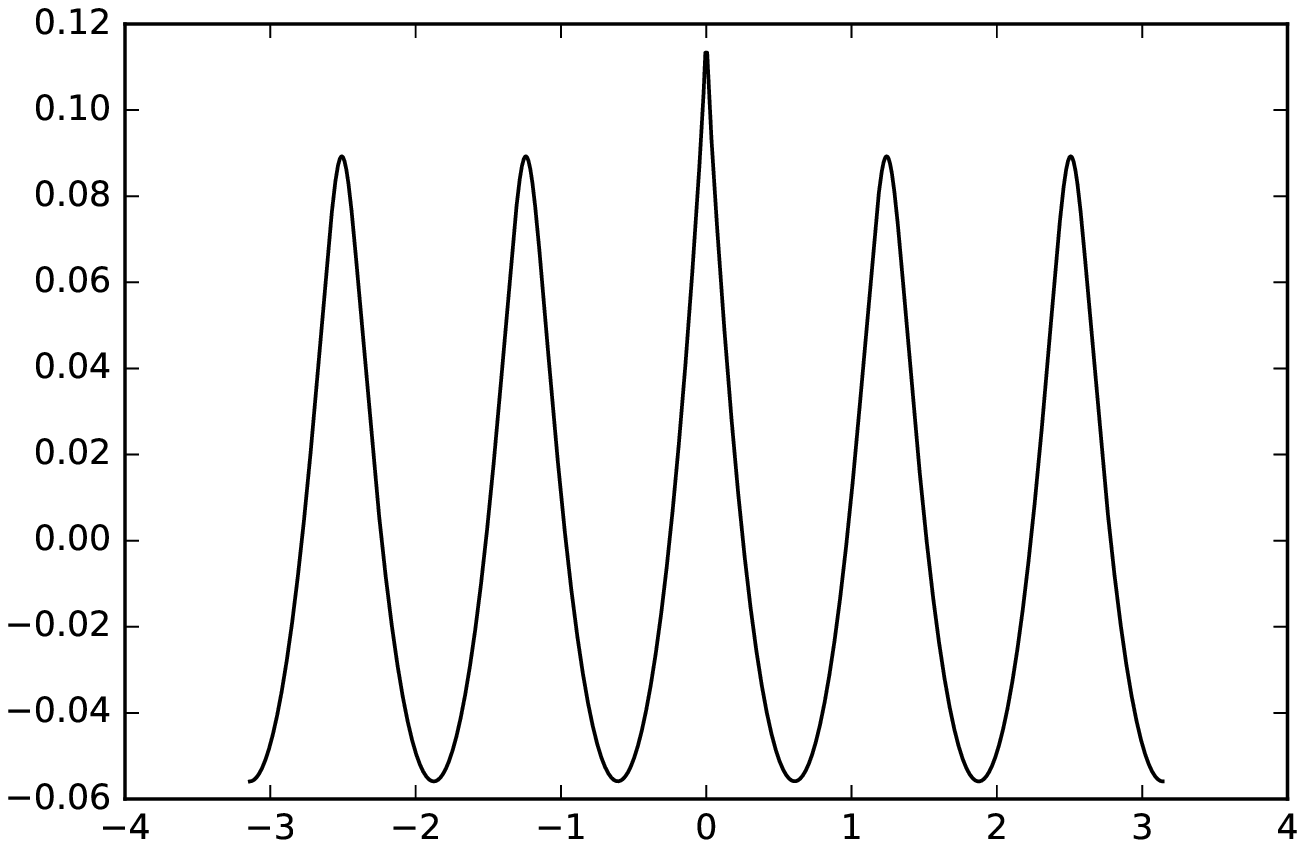}}
\end{center}
\vspace{-6mm} \caption{The profile of extreme wave corresponding to the end-point on
branch $C_{51}$ with $\mu \approx 0.22913$ and $\| w \|_\infty \approx 0.11456$.}
%\vspace{-2mm}
\end{figure}

\begin{figure}
\vspace{-3mm}
\begin{center}
\SetLabels
 \L (0.5*0.01) $x$\\
 \L (-0.005*0.48) $\eta$\\
\endSetLabels 
\leavevmode
\strut\AffixLabels{\includegraphics[width=80mm]{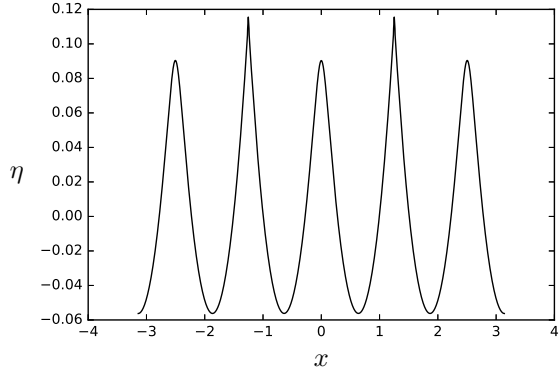}}
\end{center}
\vspace{-5mm} \caption{The profile of extreme wave corresponding to the end-point on
branch $C_{52}$ with $\mu \approx 0.23106$ and $\| w \|_\infty \approx 0.11553$. }
\vspace{-3mm}
\end{figure}

Meanwhile, the version of the software SpecTraVVave used in this paper (it is
available from the authors) includes the Navigator procedure which allows us not
only to select all branches $C_{51}$, $C_{52}$ and $C_{53}$, but also to trace them
up to the upper bound. The wave profiles plotted in Figs.~5--7 have extreme form
(each having at least one crest with the $2 \pi / 3$ angle formed by two smooth
arcs) and correspond to the end-points on each secondary branch.

A characteristic feature of wave profiles on every of these branches concerns the
number of crests per period (the latter is $2 \pi$---five times greater than the
period of waves on $C_{5}$), that are higher than the others. Indeed, only one crest
is higher than the others (which are equal) on wave profiles corresponding to
solutions on branch $C_{51}$. Two equal crests are higher than the other three
(which are again equal) on profiles corresponding to solutions on $C_{52}$. The
remaining option---three equal crests exceeding the other two which are also
equal---is realised on branch $C_{53}$.

\begin{figure}
\vspace{-2mm}
\begin{center}
\SetLabels
 \L (0.5*0.01) $x$\\
 \L (-0.005*0.48) $\eta$\\
\endSetLabels 
\leavevmode
\strut\AffixLabels{\includegraphics[width=80mm]{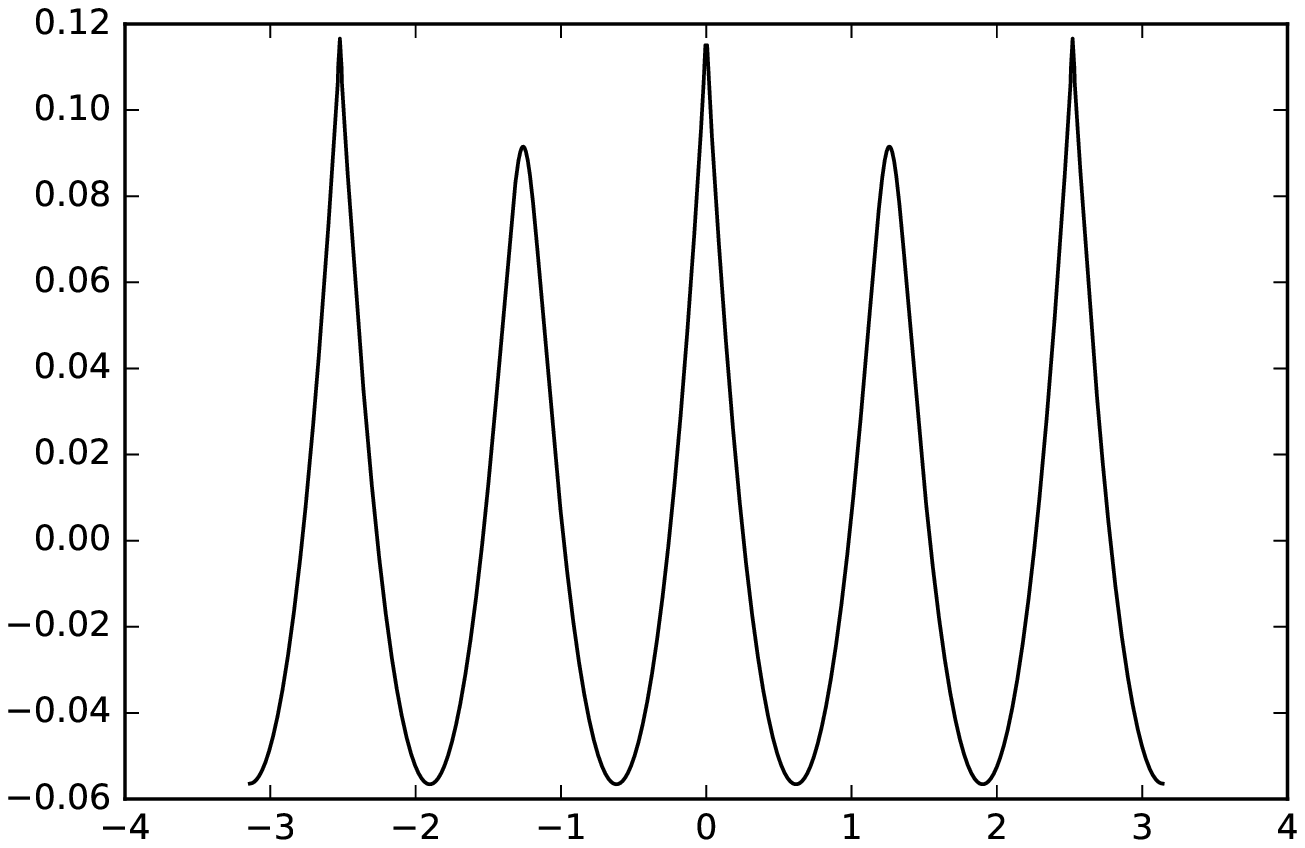}}
\end{center}
\vspace{-5mm} \caption{The profile of extreme wave corresponding to the end-point on
branch $C_{53}$ with $\mu \approx 0.23322$ and $\| w \|_\infty \approx 0.11661$. }
\vspace{-5mm}
\end{figure}

\vspace{-4mm}

\section{Concluding remarks}

We have considered the nonlinear problem describing steady, gravity waves on water
of finite depth. In our previous paper \cite{KD}, this problem had been reduced to a
single pseudo-differential operator equation \eqref{37} (Babenko's equation);
moreover, it had been demonstrated that this equation has an equivalent form
\eqref{spectralBabenko} which is more convenient for numerical solution.

Here, an operator equation (modified Babenko's equation) equivalent to
\eqref{spectralBabenko} has been derived, namely, \eqref{mod_Babenko}. Its advantage
is that the operator involved (it replaces $\mathcal{B}_r \, \D / \D t$ used in
\eqref{37} and in \eqref{spectralBabenko}) depends directly on the mean depth of
water $h$, whereas the parameter $r$ has no direct hydrodynamic meaning. On the
other hand, the new operator is nonlinear which is a drawback.

However, the algorithm developed for numerical solution of modified Babenko's
equation (it is a modification of the free software SpecTraVVave; see \cite{MVK})
allows us to tackle with this drawback very efficiently. Moreover, the developed
numerical procedure is not only very fast, but also has a remarkably high accuracy.

\vspace{2mm}

\noindent {\bf Acknowledgements.}

\noindent The authors are grateful to Henrik Kalisch without whose support the paper
would not appear. E.\,D. acknowledges the support from the Norwegian Research
Council.

\end{document}